\documentclass[sigconf]{acmart}
\usepackage{amsmath}
\usepackage{amsthm}
\numberwithin{equation}{section}

\usepackage{amssymb}
\usepackage{amsfonts}
\usepackage[utf8]{inputenc}
\usepackage[shortlabels]{enumitem}
\usepackage{hyperref}
\usepackage{natbib}
\usepackage{url}
\usepackage{color}
\usepackage{framed}
\usepackage{graphicx}
\usepackage{bm}
\usepackage{tikz-cd}
\usepackage{xcolor}  
\usepackage{float}
\usepackage{mathtools}

\usepackage{graphicx}
\title{A Noncommutative Nullstellensatz for Perfect Two-Answer Quantum Nonlocal Games}

\author{Tianshi Yu}
\affiliation{
	\institution{Key Lab of Mathematics Mechanization, AMSS }
	\department{University of Chinese Academy of Sciences}
	\city{Beijing, 100190}
	\country{China}
}
\email{yutianshi@amss.ac.cn}

\author{Lihong Zhi}
\affiliation{
	\institution{Key Lab of Mathematics Mechanization, AMSS }
	\department{ University of Chinese Academy of Sciences}
	\city{Beijing, 100190}
	\country{China}
}
\email{lzhi@mmrc.iss.ac.cn}

\copyrightyear{2025}
\acmYear{2025}
\setcopyright{rightsretained}
\acmConference[ISSAC 2025]{International Symposium on Symbolic and Algebraic Computation 2023}{July 28th--August 1st, 2025}{Guanajuato, Mexico}
\acmBooktitle{International Symposium on Symbolic and Algebraic Computation 2025 (ISSAC 2025), July 28th--August 1st, 2025, Guanajuato, Mexico}\acmDOI{10.1145/xxxxxxx.xxxxxxx}
\acmISBN{xxxxxx}


\begin{document}

\newtheorem{defi}{Definition}[section]
\newtheorem{thm}{Theorem}[section]
\newtheorem{cor}[thm]{Corollary}
\newtheorem{prop}[thm]{Proposition}
\newtheorem{claim}[thm]{Claim}
\newtheorem{remark}{Remark}
\newtheorem{Lemma}{Lemma}
\newtheorem{Example}{Example}
\newcommand{\bN} { {\mathbb{N}}}   
\newcommand{\bC} { {\mathbb{C}}}  
\newcommand{\bQ} { {\mathbb{Q}}}   
\newcommand{\bZ} { {\mathbb{Z}}}   
\newcommand{\bR} { {\mathbb{R}}}   
\newcommand{\bF} { {\mathbb{F}}}
\newcommand{\bH} { {\mathbb{H}}}
\newcommand{\bK} { {\mathbb{K}}}
\newcommand{\op} { {\mathrm{op}}}
\newcommand{\alg} { {\mathrm{alg}}}
\newcommand{\SOS} { {\mathrm{SOS}}}
\newcommand*\abs[1]{\left\lvert#1\right\rvert} 
\newcommand*\norm[1]{\left\lVert#1\right\rVert} 

\let\bar=\overline
\let\leq=\leqslant
\let\geq=\geqslant

\begin{abstract}
	This paper introduces a noncommutative version of the Nullstellensatz, motivated by the study of quantum nonlocal games. It has been proved that a two-answer nonlocal game with a perfect quantum strategy also admits a perfect classical strategy. 
	We generalize this result to the infinite-dimensional case, showing that a two-answer game with a perfect commuting operator strategy also admits a perfect classical strategy. This result induces a special case of noncommutative Nullstellensatz.
\end{abstract}

\maketitle
\section{Introduction}


Quantum nonlocal games have been a vibrant area of research across mathematics, physics, and computer science in recent decades. They help understand quantum nonlocality, which was famously verified by the violation of Bell inequalities \cite{bell1964einstein}. 
In 1969, Clauser et al. first introduced quantum nonlocal games \cite{clauser1969proposed}. 
A nonlocal game typically involves two or more players and a verifier. The verifier sends questions to the players independently, and each player responds without any communication between them. A predefined scoring function determines whether the players win based on the given questions and their answers. 
The distinction between classical and quantum strategies lies in whether players can share quantum entanglement. For instance, in the CHSH game, the classical strategy limits the winning probability to at most $ \frac{3}{4} $. In contrast, quantum strategies using shared entangled states can achieve a success probability of $ \cos^2(\frac{\pi}{8})\approx 0.85 $.


The mathematical models of quantum nonlocal games are often described using algebraic structures \cite{dykema2019non,watts2021noncommutative,watts20203xor,lupini2020perfect}. 
$*-$algebras, noncommutative Nullstellensatz (see \cite{cimprivc2013noncommutative,cimprivc2014real,cimpric2015real}) and Positivstellensatz (see \cite{H2002, HMS2004}) are used for characterizing the different types of strategies for nonlocal games.  Our previous work also gave an algebraic characterization for perfect strategies of mirror games using the universal game algebra, Nullstellensatz, and sums of squares~\cite{yan2023characterization}.

 Nonlocal games with two answers are games in which the set of possible responses consists of only two options~\cite{watts2021noncommutative} (also called binary games in \cite{cleve2004consequences}).  This paper proposes a noncommutative Nullstellensatz inspired by the perfect commuting operator strategies for two-answer nonlocal games. Specifically, we proved that a two-answer game that admits a perfect commuting operator strategy also has a perfect classical strategy, a generalization of the work \cite[Theorem 3]{cleve2004consequences}. Combined with the algebraic characterization of perfect commuting operator strategy \cite{watts2021noncommutative}, we get a new form of noncommutative Nullstellensatz. Although our problem is motivated by nonlocal games, our proofs are presented in a purely algebraic form, allowing readers unfamiliar with quantum nonlocal games to directly engage with the algebraic versions of the theorems.

\section{Preliminaries}
\subsection{Motivations}
If the readers are familiar with this field, they can skip the content of this subsection.

A quantum nonlocal game $\mathcal{G}$ can be described as a scoring function $ \lambda $ from the finite set $ X\times Y\times A \times B $ to $ \{0,1\} $, where the player Alice has a question set $ X $ and an answer set $ A $, while the player Bob has a question set $ Y $ and an answer set $ B $. In a round of the game, Alice would receive the question $x\in X$ and answer $a\in A$ according to $x$ and her strategy; similarly, Bob would receive the question $y\in Y$ and answer $b\in B$.
The players cannot communicate during the game but can make arrangements before playing it.
The players are considered to win the game when  
$ \lambda(x,y,a,b)=1$, and they lose in all other cases.
    
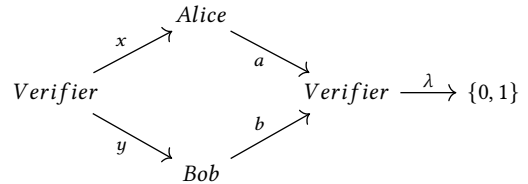
\begin{center}
    \begin{tikzcd}
		& Alice \arrow[rd, "a"'] &                               &           \\
		Verifier \arrow[ru, "x"] \arrow[rd, "y"'] &                       & Verifier \arrow[r, "\lambda"] & {\{0,1\}} \\
		& Bob \arrow[ru, "b"]  &                               &          
	\end{tikzcd}
    $$\lambda(x,y, a,b)=\left\{ \begin{array}{ll}
    1   ~~~~~ & \text{win} \\
    0   ~~~~~ & \text{lose}
    \end{array} \right. $$
    \end{center}
    
A (deterministic) {\it classical strategy} involves two mappings 
$$ u: X\rightarrow A \text{~~and~~} v: Y\rightarrow B; $$ 
when Alice receives a question $ x\in X $, she responds with $ u(x) $, and similarly, Bob responds with $v(y)$ when he receives $y\in Y$.

If the players share a quantum state $ \phi $ on a (perhaps infinite-dimensional) Hilbert space $ \mathcal{H} $, 
and for every question pair $ (x,y)\in X\times Y$, Alice and Bob perform commuting projection-valued measurements (PVMs) 
$$ \left\{E_a^x\in\mathcal{B}(\mathcal{H}):\sum_{a\in A}E_a^x=\bm{1}\right\}  \text{~and~} \left\{F_b^y\in\mathcal{B}(\mathcal{H}):\sum_{b\in B}F_b^y=\bm{1}\right\} $$ respectively to determine their answers,
then the game is said to have a {\it commuting operator strategy}. 
    \begin{equation*}
    \begin{aligned} x\longrightarrow \text{Alice} &\xrightarrow{\{E^{x}_{a_i},~a_i\in A\}}\phi\in \mathcal{H} \longrightarrow a\\
	y\longrightarrow \text{Bob} &\xrightarrow{\{F^{y}_{b_j},~b_j\in B\}}\phi  \in \mathcal{H}\longrightarrow b\end{aligned}
    \end{equation*}
The PVMs satisfy the following relations:
\begin{align*}
    &E_a^xF_b^y-F_b^yE_a^x=0,~\forall (x,y,a,b)\in X\times Y\times A \times B;\\
    &(E_a^x)^2=E_a^x=(E_a^x)^*,~\forall x\in X,a\in A;\\
    &(F_b^y)^2=F_b^y=(F_b^y)^*,~\forall y\in Y,b\in B;\\
    &E_{a_1}^xE_{a_2}^x=0,~\forall x\in X, a_1\neq a_2\in A;\\
    &F_{b_1}^yF_{b_2}^y=0,~\forall y\in Y, b_1\neq b_2\in B;\\
    &\sum_{a\in A}E_a^x=\bm{1},~\forall x\in X;\\
    &\sum_{b\in B}F_b^y=\bm{1},~\forall y\in Y.
\end{align*}
These relations can be abstracted to obtain the universal game algebra for the nonlocal game $\mathcal{G}$ \cite[Section 3]{watts2021noncommutative}.

Furthermore, if we restrict the quantum state $ \phi $ to be a tensor $ \phi_1\otimes\phi_2 $, where $ \phi_1 $ and $ \phi_2 $ are in finite-dimensional Hilbert space $ \mathcal{H}_1 $ and $ \mathcal{H}_2 $ respectively, then we get a (finite-dimensional) {\it quantum strategy}.

 By defining the three types of strategies, we know the classical strategies are contained in the quantum strategies, which are included in the commuting operator strategies. We call a {\it strategy perfect} if the players can always win the game using this strategy. 
 Therefore, a game that admits a perfect classical strategy also has a perfect commuting operator strategy.
However, the converse does not hold. For example, the famous Magic Square game admits a perfect quantum strategy but has no perfect classical strategy \cite{cleve2004consequences}. However, in certain exceptional cases, these strategies may be equivalent. 

For a two-answer game, that is, one whose answer sets are both  $\{0, 1\}$, if it admits a perfect quantum strategy, then Cleve,  Hoyer, Toner, and  Watrous showed that the two-answer game must have a perfect classical strategy \cite[Theorem 3]{cleve2004consequences}. 
We contribute to extending this theorem to the infinite-dimensional case, proving that a two-answer game with a perfect commuting operator strategy also admits a perfect classical strategy. This result, combined with the work of Watts, Helton, and Klep \cite[Theorem 4.3]{watts2021noncommutative}, 
derive a version of the noncommutative Nullstellensatz using a sum of squares (SOS) expression.

\subsection{
Universal Game Algebra for Two-Answer Games}\label{Def}
Let $X,Y,A,B$ be finite sets, where $A=B=\{0,1\}$, and $ \bC\langle \{e_a^x,f_b^y\}\rangle $ be the free algebra generated by $ \{e_a^x,f_b^y:~(x,y,a,b)\in X\times Y\times A\times B\} $.

Define the two-sided ideal
\begin{align*}
\mathcal{I}=\langle&(e_a^x)^2-e_a^x,~(f_b^y)^2-f_b^y;\\
&\sum_{a\in A}e_a^x-1,~\sum_{b\in B}f_b^y-1;\\
&e_a^xf_b^y-f_b^ye_a^x \mid x\in X,y\in Y,a\in A,b\in B\rangle
\end{align*}
and let 
\begin{align}\label{defA}
\mathcal{A}=\bC\langle \{e_a^x,f_b^y\}\rangle/\mathcal{I}.
\end{align}
Since
\begin{equation*}
\begin{aligned}
e_0^xe_1^x&=\frac{1}{2}\Big(\left(e_0^x+e_1^x-1\right)^2-\big((e_0^x)^2-e_0^x\big)\\
&\quad-\big((e_1^x)^2-e_1^x\big)+\left(e_0^x+e_1^x-1\right)\Big),
\end{aligned}
\end{equation*}
we have 
$$ e_0^xe_1^x\in\mathcal{I},~\forall x\in X.$$
Similarly, one can show that 
$$f_0^yf_1^y\in\mathcal{I},~\forall y\in Y.$$ 

The elements in $ \mathcal{I} $ are the relationships the generators satisfy. 
We can also equip $ \mathcal{A} $ with the natural involution $ "*" $ induced by \[(e_a^x)^*=e_a^x, \,\, (f_b^y)^*=f_b^y. \] 
Then $\mathcal{A}$ is a complex $*-$algebra.

The relations in $\mathcal{A}$ are precisely those satisfied by the PVMs of a two-answer game. Thus, this algebra can characterize the commuting operator strategies of a two-answer game. $\mathcal{A}$ serves as the {\it universal game algebra} for two-answer games, as discussed in \cite[Section 3]{watts2021noncommutative}. Furthermore, $ \mathcal{A} $ is a group algebra. 

Let 
\begin{align}\label{AxBy}
     A_x=e_0^x-e_1^x,~B_y=f_0^y-f_1^y 
\end{align}
for any $ x\in X,~y\in Y $, we have 
\begin{align}
&A_x^2=B_y^2=1,A_x=A_x^*,~B_y=B_y^*,\\&e_a^x=\dfrac{1+(-1)^a A_x}{2},~f_b^y=\dfrac{1+(-1)^b B_y}{2}.
\end{align}
Let  $G$ be the group generated by elements $ A_x,~x\in X $ and $ B_y,~y\in Y $. Equip the group algebra of $G$ with the natural involution $*$:
\[g^*=g^{-1}, 
~~(g_1g_2)^*=g_2^*g_1^*,~\forall g,g_1,g_2\in G,\]
Then, we observe that 
 \[\mathcal{A}=\bC[G].\] 

We denote 
\begin{equation*}
\SOS_{\mathcal{A}}:=\left\{\sum_{i=1}^{n}\alpha_i^*\alpha_i\mid n\in\bN,~\alpha_i\in\mathcal{A}\right\}.
\end{equation*}
It is well known that $ \SOS_{\mathcal{A}} $ is Archimedean (see  \cite[example 3]{cimprivc2009representation} or \cite[Remark 4.1]{netzer2013real}), that is,  for every $ \alpha\in\mathcal{A}$, it  can be shown that 
 \[\|a\|_1^2\ -\alpha^*\alpha\in\SOS_{\mathcal{A}},\]
where $\|a\|_1=\sum_{g\in G} |a_g|$. 



We  also need to introduce the concept of
$*-$representation.

\begin{defi}
    A $*$-representation of $\mathcal{A}$ is a  unital $*-$homomorphism
    \begin{equation*}
\sigma:\mathcal{A}\rightarrow\mathcal{B}(\mathcal{H}),
    \end{equation*}
    where $ \mathcal{B}(\mathcal{H}) $ denotes the set of bounded linear operators on a Hilbert space $ \mathcal{H} $ and $ \sigma $ satisfies $ \sigma(u^*)=\sigma(u)^*, \forall u\in \mathcal{A} $.
\end{defi}

\section{Main Results}

Let $X,Y,A,B$ be finite sets, where $A=B=\{0,1\}$, and $ \bC\langle \{e_a^x,f_b^y\}\rangle $ be the free algebra generated by $ \{e_a^x,f_b^y:~(x,y,a,b)\in X\times Y\times A\times B\}$.
Let  $ \mathcal{A} $ be the complex $*-$algebra defined in the previous Subsection \ref{Def}. 
Our main result is stated below: 

\begin{thm}\label{thm1}
	Let $ \mathcal{A} $ denote the universal game algebra for two-answer games.  
	Let
    \begin{equation}\label{Lambda}
         \Lambda\subseteq X\times Y\times A\times B 
    \end{equation}
    be the index set of $\mathcal{N}$, where 
    \begin{equation}\label{N}
    \mathcal{N}=\{e_a^xf_b^a\mid(x,y,a,b)\in\Lambda\}.
    \end{equation}
    Let  $\mathcal{L}(\mathcal{N})$
    be the left ideal  generated by $\mathcal{N}$. Then
	\begin{align*}
	&\qquad-1\notin\operatorname{SOS}_{\mathcal{A}}+\mathcal{L}(\mathcal{N})+\mathcal{L}(\mathcal{N})^*\end{align*}
    if and only if there exists a $*-$representation \[\rho:\mathcal{A}\rightarrow\bC\]
    such that 
    \[\rho(\mathcal{N})=\{0\}.\]
    
\end{thm}

\begin{remark}\label{rmk1}
    We can interpret the set $\mathcal{N}$ as the invalid determining set of a two-answer game, where the scoring function 
    \begin{equation*}
    \lambda(x,y,a,b)=0 ~~{\text{if}}~~ e_a^xf_b^a\in\mathcal{N}.
    \end{equation*}
See~\cite[Definition 3.4]{watts2021noncommutative}.
\end{remark}

We prove this theorem by the following propositions.

\begin{prop}\label{p1}(\cite[Theorem 4.3]{watts2021noncommutative})
    Let $ \mathcal{A} $ denote  the universal game algebra for two-answer games. If 
	\begin{equation*}
-1\notin\operatorname{SOS}_{\mathcal{A}}+\mathcal{L}(\mathcal{N})+\mathcal{L}(\mathcal{N})^*,
	\end{equation*}
	there exists a $ *-$representation
    \[
\sigma:\mathcal{A}\rightarrow\mathcal{B}(\mathcal{H}) \]
and $ 0\neq\psi\in \mathcal{H}$, where $ \mathcal{H} $ is a separable Hilbert space, such that 
	\begin{equation*}
	\sigma(\alpha)\psi=0
	\end{equation*}
	for all $ \alpha\in\mathcal{L}(\mathcal{N}) $.
\end{prop}

We emphasize that $ \mathcal{H} $ is a separable Hilbert space, which will be used in the proof of Proposition \ref{p2}. For completeness, we briefly outline the proof given by  Watts, Helton, and Klep in \cite[Theorem 4.3]{watts2021noncommutative}.

\begin{proof}[Proof Sketch]
	By the Hahn-Banach theorem \cite[Theorem III.1.7]{barvinok2002course} and Archimedeanity of $ \SOS_{\mathcal{A}} $, there exists a functional $ f:\mathcal{A}\rightarrow\bC $ which strictly separate $ -1 $ and $ \SOS_{\mathcal{A}}+\mathcal{L}(\mathcal{N})+\mathcal{L}(\mathcal{N})^* $, i.e 
	\begin{equation*}
	f(-1)=-1,~f(\SOS_{\mathcal{A}}+\mathcal{L}(\mathcal{N})+\mathcal{L}(\mathcal{N})^*)\subseteq\bR_{\geq 0}.
	\end{equation*}
    We list the properties of $f$ as follows:
    \begin{itemize}
        \item $f(\mathcal{L}(\mathcal{N}))=\{0\}\text{~and~}f(\SOS_{\mathcal{A}})\subseteq\bR_{\geq 0}.$
        \item  $f(h^*)=f(h)^*$ for every $h\in\mathcal{A}$.
    \end{itemize}
	
	
	Now, the GNS construction yields the desired *-representation $ \sigma $ and a cyclic vector $\psi$. Define the sesquilinear 
	form on $\mathcal{A}$
	\begin{equation*}
	\langle \alpha\mid \beta\rangle=f(\beta^*\alpha),
	\end{equation*}
	and 
    \begin{equation}
    M=\{\alpha\in\mathcal{A}:~f(\alpha^*\alpha)=0\}.
    \end{equation}
    By Cauchy-Schwarz inequality, $M$ is a left ideal of $\mathcal{A}$. 
	Form the quotient space $\widetilde{\mathcal{H}}:=\mathcal{A}/M$, and equip it with the inner product $\langle\cdot\mid\cdot\rangle$. We can complete $ \widetilde{\mathcal{H}} $ to the Hilbert space $\mathcal{H}$. 
	
	It is worth mentioning that we can assume $ \mathcal{H} $  to be a separable Hilbert space, as this assumption holds because $ \mathcal{A} $ 
	has only a finite number of generators, allowing us to generate a countable dense subset of $ \mathcal{A} $ using these generators with rational coefficients. 
    By applying this to the quotient space, we establish the separability of
 $\mathcal{H}$.
	
	Define the quotient map \begin{align*}\phi:\mathcal{A}&\rightarrow\mathcal{H} \\
    \alpha&\mapsto \alpha+M,
    \end{align*}
    the cyclic vector \[\psi:=\phi(1)=1+M,\]
    and the left regular representation 
	\begin{align*}
\sigma:\mathcal{A}&\rightarrow\mathcal{B}(\mathcal{H})\\
\alpha&\mapsto\left(p+M\mapsto \alpha p+M\right).
	\end{align*}
	By Archimedeanity of $ \SOS_{\mathcal{A}} $, it is easy to verify that $ \sigma(\alpha) $ is bounded for every $ \alpha\in\mathcal{A}$, and thus $ \sigma $ is a $*-$representation. Finally, the result
    \[\sigma(\mathcal{L}(\mathcal{N}))\psi=\{0\} \]
    follows from 
    $$ \mathcal{L}(\mathcal{N})^*\mathcal{L}(\mathcal{N})\subseteq\mathcal{L}(\mathcal{N})\subseteq M. $$

\end{proof}

\begin{prop}\label{p2}
    Let $ \mathcal{A} $ denote the universal game algebra for two-answer games. 
    Suppose there exists a $ *-$representation 
    \[\sigma:\mathcal{A}\rightarrow\mathcal{B}(\mathcal{H}),\]
    and $ 0\neq\psi\in \mathcal{H} $, where $ \mathcal{H} $ is a separable Hilbert space, such that 
	\begin{equation*}
	\sigma(\alpha)\psi=0
	\end{equation*}
	for all $ \alpha\in\mathcal{L}(\mathcal{N}) $ ($\mathcal{N}$ is defined in equation (\ref{N})). Then there exists a one-dimensional $*-$representation $\rho:\mathcal{A}\rightarrow\bC$ such that $$ \rho(\mathcal{N})=\{0\}. $$
\end{prop}

\begin{remark}
 The proof below extends the argument in 
 \cite[Theorem 4]{cleve2004consequences}, which was originally stated for the tensor product of two finite-dimensional Hilbert spaces, to the more general setting of infinite-dimensional Hilbert spaces.
 In fact, the condition in Proposition \ref{p2} implies that the tuple 
 $$(\mathcal{H},\{\sigma(e_a^x)\},\{\sigma(f_b^y)\},\psi)$$ defines a perfect commuting operator strategy for the two-answer game with an invalid determining set $\mathcal{N}$. Furthermore,   the conclusion of Proposition \ref{p2} demonstrates  that the mappings 
 induced by $\rho$: 
 \[x\mapsto a~(\text{satisfying~}\rho(e_a^x)=1),  y \mapsto b~(\text{satisfying~}\rho(f_b^y)=1)\]
 for $x \in X, y \in Y$,  provide a perfect deterministic strategy for this game. Therefore, this proposition also implies that a two-answer game with a perfect commuting operator strategy also admits a perfect classical strategy.
 \end{remark}

\begin{proof}
	We construct the one-dimensional representation $ \rho $ as follows. Since 
	\begin{equation*}
	\sum_{a\in A}\sum_{b\in B}\psi^*\sigma(e_a^x f_b^y)\psi=1
	\end{equation*}
	for every fixed pair $ (x,y) $, we know that there exist $ (x,y,a,b)\in X\times Y\times A\times B $ such that 
    \[\psi^*\sigma(e_a^x f_b^y)\psi\neq 0 .\]
    Let 
	\begin{equation}\label{SetPi}
	\Pi=\{(x,y,a,b)\in X\times Y\times A\times B:~\psi^*\sigma(e_a^x f_b^y)\psi\neq 0\},
	\end{equation}
	 we have 
     \[\Pi\subseteq X\times Y\times A\times B\setminus\Lambda \]
     since 
     \[\sigma(\mathcal{L}(\mathcal{N}))\psi=\{0\},\]
     and thus 
     \[\psi^*\sigma(e_a^x f_b^y)\psi=0\]
     for any $ (x,y,a,b)\in\Lambda $, where  $\Lambda$ is the index of the invalid determining set $\mathcal{N}$ (\ref{N}), see Remark \ref{rmk1}. 
	
	Using the generators $ A_x $ and $ B_y $ defined in  (\ref{AxBy}), we can rewrite:
	\begin{equation}\label{eq1}
	\begin{aligned}
	\psi^*\sigma(e_a^x f_b^y)\psi&=\dfrac{1}{4}\\
	&+\dfrac{1}{4}(-1)^a\psi^*\sigma(A_x)\psi\\
	&+\dfrac{1}{4}(-1)^b\psi^*\sigma(B_y)\psi\\
	&+\dfrac{1}{4}(-1)^{a+b}\psi^*\sigma(A_xB_y)\psi.
	\end{aligned}
	\end{equation}
	
	Since $ \mathcal{H} $ is separable, we can choose an orthogonal basis of $ \mathcal{H} $ named
	\begin{equation*}
	\{\psi_1,\psi_2,\ldots\},
	\end{equation*}
	where $ \psi_1=\psi $. Define
	\begin{align*}
	k:X&\rightarrow\bN\\x&\mapsto\min\{j\in\bN:~\psi_j^*\sigma(A_x)\psi\neq 0\};\\
	l:Y&\rightarrow\bN\\y&\mapsto\min\{j\in\bN:~\psi_j^*\sigma(B_y)\psi\neq 0\}.
	\end{align*}
    Unlike the finite-dimensional case considered in the proof of \cite[Theorem 4]{cleve2004consequences}, here we need to show that 
 for every $ x\in X $, $y \in Y $,  $ k(x)$ and $l(y)$ are well-defined. 
 
 As $ \psi\neq 0 $ and $ \sigma(A_x)^2=1$,  there must exist a $ j\in \bN $ such that $ \psi_j^*\sigma(A_x)\psi\neq 0 $ (otherwise, $ \sigma(A_x)\psi=0$, which contradicts the assumption that  $ \psi\neq 0 $ and $ \sigma(A_x)^2=1$).  
    Similarly, we can also prove that $ l(y) $ is well-defined. 
	
	Let 
	\begin{align}
	u:X&\rightarrow A\nonumber\\x&\mapsto\left\{
	\begin{aligned}
	&0,~0\leq\arg \psi_{k(x)}\sigma(A_x)\psi<\pi;\\&1,~\pi\leq\arg \psi_{k(x)}\sigma(A_x)\psi<2\pi.
	\end{aligned}\right.\label{u}\\
    v:Y&\rightarrow B\nonumber\\y&\mapsto\left\{
	\begin{aligned}
	&0,~0\leq\arg \psi_{l(y)}\sigma(B_y)\psi<\pi;\\&1,~\pi\leq\arg \psi_{l(y)}\sigma(B_y)\psi<2\pi.
	\end{aligned}\right.\label{v}
    \end{align}
	
	We have the following claim:
	\begin{claim}\label{c1}
		For every $ (x,y,u(x),v(y))\in X\times Y\times A\times B $, we have 
        \[(x,y,u(x),v(y))\in \Pi, \]
        where $\Pi$ is defined in (\ref{SetPi}). 
        In particular, this implies
        \[\psi^*\sigma(e_{u(x)}^{x} f_{v(y)}^{y})\psi\neq 0.\]
	\end{claim}
	
	We will provide the proof of Claim \ref{c1} after completing the proof of Proposition \ref{p2}. 
    
    We construct the one-dimensional $*-$representation $\rho $ as follows. For every $ x\in X $, define
	\begin{equation*}
	\rho(e_{u(x)}^{x})=1,~\rho(e_{1-u(x)}^{x})=0;
	\end{equation*}
	and for every $ y\in Y $, define
	\begin{equation*}
	\rho(f_{v(y)}^{y})=1,~\rho(f_{1-v(y)}^{y})=0.
	\end{equation*}
	Then, by linearity and homogeneity, we extend $ \rho $ to the entire game algebra  $ \mathcal{A}$.
    
    Since $\rho(e_a^x)$ and $\rho(f_b^y)$ are either $0$ or $1$, they are naturally commutative.  It is straightforward to check that  $\rho$ satisfies: 
   
    \[\rho(e_a^x)^2=\rho(e_a^x),~ \rho(f_b^y)^2=\rho(f_b^y),\]
    and
    \[\rho(e_0^x)+\rho(e_1^x)=1,~\rho(f_0^y)+\rho(f_1^y)=1,\]
     for all $e_a^x$, $f_b^y\in \mathcal{A}$, i.e.,  $ \rho(e_a^x) $ and $ \rho(f_b^y) $ satisfy the same relations as $e_a^x$ and $f_b^y$ in $ \mathcal{A}$. 
    Therefore, $ \rho $ is indeed a $*-$representation of $\mathcal{A}$.
    
    Since
	\begin{equation*}
	\rho(e_a^xf_b^y)=1\iff \left(a=u(x)\right)\wedge \left(b=v(y)\right).
	\end{equation*}
	 By Claim \ref{c1}, we have 
     \begin{equation}\label{pi1}
     \rho(e_a^xf_b^y)=1\Longrightarrow (x,y,a,b)\in \Pi.
     \end{equation}
    Since the value of $\rho(e_a^xf_b^y)$ can only be  $1$ or $0$,  as \[ \Pi\cap\Lambda=\emptyset,\]
   the condition (\ref{pi1}) implies that for 
    every $ (x,y,a,b)\in\Lambda$, i.e., for every 
	$ e_a^xf_b^y\in\mathcal{N}$, 
    \[\rho(e_a^xf_b^y)=0\]
    holds, which completes the proof.
\end{proof}
\begin{remark}
   For quantum nonlocal games, the value $\psi^*\sigma(e_a^xf_b^y)\psi$ is the probability that the players provide the answer pair $(a,b)$ for the question pair $(x,y)$. Since we start with a perfect commuting operator strategy, $\psi^*\sigma(e_a^xf_b^y)\psi\neq 0$ implies that the scoring function $\lambda(x,y,a,b)=1$. In other words, Claim \ref{c1} indicates that $$\lambda(x,y,u(x),v(y))=1.$$ 
   That is, the mappings $u: X\rightarrow A$ and $v: Y\rightarrow B$ defined in equations (\ref{u}) and (\ref{v}) actually define a perfect deterministic classic strategy for the two-answer game.
\end{remark}
Now, we provide the proof of Claim \ref{c1}.
\begin{proof}[Proof of Claim \ref{c1}]
	We set $ a=u(x) $ and $ b=v(y) $ in equations (\ref{eq1}), and compute 
	\begin{equation}\label{eq2}
	\begin{aligned}
	\psi^*\sigma(e_{u(x)}^{x} f_{v(y)}^{y})\psi&=\dfrac{1}{4}\\
	&+\dfrac{1}{4}(-1)^{u(x)}\psi^*\sigma(A_x)\psi\\
	&+\dfrac{1}{4}(-1)^{v(y)}\psi^*\sigma(B_y)\psi\\
	&+\dfrac{1}{4}(-1)^{u(x)+v(y)}\psi^*\sigma(A_xB_y)\psi.
	\end{aligned}
	\end{equation}
	Notice that $ \sigma(A_x) $ and $ \sigma(B_y) $ are commutative self-adjoint operators,  so $ \psi^*\sigma(A_x)\psi,~\psi^*\sigma(B_y)\psi $ and $ \psi^*\sigma(A_xB_y)\psi $ are all real numbers.
	
	If $ \psi^*\sigma(A_x)\psi\neq 0 $, and since $ \psi_1=\psi$,  we conclude that 
    $ k(x)=1 $. 
    Moreover, according to (\ref{u}), 
    if $\psi^*\sigma(A_x)\psi>0$, we have 
    \[u(x)=0,~(-1)^{u(x)}=1;\]
    if $\psi^*\sigma(A_x)\psi<0$, we have 
    \[u(x)=1,~(-1)^{u(x)}=-1.\]
    Therefore, the value below is always positive:
	$$ (-1)^{u(x)}\psi^*\sigma(A_x)\psi>0. $$ 
    Similarly, if $ \psi^*\sigma(B_y)\psi\neq 0 $, we have $$ (-1)^{v(y)}\psi^*\sigma(B_y)\psi>0. $$ 
	Therefore, if either $ \psi^*\sigma(A_x)\psi $ or $ \psi^*\sigma(B_y)\psi $ is nonzero,  we have 
	\begin{equation*}
	\dfrac{1}{4}(-1)^{u(x)}\psi^*\sigma(A_x)\psi+\dfrac{1}{4}(-1)^{v(y)}\psi^*\sigma(B_y)\psi>0.
	\end{equation*}
	Since $ \dfrac{1}{4}+\dfrac{1}{4}(-1)^{u(x)+v(y)}\psi^*\sigma(A_xB_y)\psi\geq 0 $, we have $$ \psi^*\sigma(e_{u(x)}^{x} f_{v(y)}^{y})\psi>0. $$
	
Hence,  we only need to consider the case $$ \psi^*\sigma(A_x)\psi=\psi^*\sigma(B_y)\psi=0. $$ 
Since we are working with infinite-dimensional separable Hilbert spaces, we modify the argument in \cite[Theorem 4]{cleve2004consequences} by incorporating Cauchy-Schwarz inequality and Parseval's identity for proving  that  
    \[\dfrac{1}{4}+\dfrac{1}{4}(-1)^{u(x)+v(y)}\psi^*\sigma(A_xB_y)\psi> 0.\]
	Conversely, suppose 
    \begin{equation}\label{negative}
        (-1)^{u(x)+v(y)}\psi^*\sigma(A_xB_y)\psi=-1
    \end{equation}
    holds. 
	By Cauchy-Schwarz's inequality, we know that 
	\begin{equation*}
	\begin{aligned}
	&\quad\left|(-1)^{u(x)+v(y)}\psi^*\sigma(A_xB_y)\psi\right|\\&\leq\|(-1)^{u(x)}\sigma(A_x)\psi\|\cdot\|(-1)^{v(y)}\sigma(B_y)\psi\|.
	\end{aligned}
	\end{equation*}
	Since $ \psi $ is a unit vector and the eigenvalues of $ \sigma(A_x),\sigma(B_y) $ can only be $ \pm 1 $, we know $$ \|(-1)^{u(x)}\sigma(A_x)\psi\|=1 ~\text{and}~ \|(-1)^{v(y)}\sigma(B_y)\psi\|=1. $$
    Applying the equality condition of the Cauchy-Schwarz inequality, and the assumption  (\ref{negative}), we obtain
	\begin{equation}\label{eq3}
	(-1)^{u(x)}\sigma(A_x)\psi=-(-1)^{v(y)}\sigma(B_y)\psi.
	\end{equation}

	By Parseval's identity, we have
	\begin{equation*}
(-1)^{u(x)}\sigma(A_x)\psi=\sum_{j=1}^{\infty}(-1)^{u(x)}\langle\sigma(A_x)\psi,\psi_j\rangle\cdot\psi_j,
	\end{equation*}
	and
	\begin{equation*}
(-1)^{v(y)}\sigma(B_y)\psi=\sum_{j=1}^{\infty}(-1)^{v(y)}\langle\sigma(B_y)\psi,\psi_j\rangle\cdot\psi_j,
	\end{equation*}
	Then equation (\ref{eq3}) gives  
	\begin{equation*}
(-1)^{u(x)}\langle\sigma(A_x)\psi,\psi_j\rangle=-(-1)^{v(y)}\langle\sigma(B_y)\psi,\psi_j\rangle,
	\end{equation*}
	which implies that 
	\begin{equation}\label{eq4}
	(-1)^{u(x)}\psi_j^*\sigma(A_x)\psi=-(-1)^{v(y)}\psi_j^*\sigma(B_y)\psi
	\end{equation}
	holds for every $ j\in\{1,2,\ldots\ldots\} $. 
	However, equation (\ref{eq4}) must fail to hold for $ j=\min\{k(x),l(y)\} $.  It is clear that equation (\ref{eq4}) fails when $ k(x)\neq l(y)$.  Assume  $ k(x)=l(y)=j$, 
	we find that the arguments of 
 $ \arg\left((-1)^{u(x)}\psi_j^*\sigma(A_x)\psi\right) $ and 
	$ \arg\left((-1)^{v(y)}\psi_j^*\sigma(B_y)\psi\right) $ both lie in the range  $ \left[0,\pi\right) $, which contradicts equation   (\ref{eq4}) once again!
	
	Therefore, when $ \psi^*\sigma(A_x)\psi=\psi^*\sigma(B_y)\psi=0$,  we have shown that \[\dfrac{1}{4}+\dfrac{1}{4}(-1)^{u(x)+v(y)}\psi^*\sigma(A_xB_y)\psi> 0.\]
	That is, $$ \psi^*\sigma(e_{u(x)}^{x} f_{v(y)}^{y})\psi>0, $$ which always holds, thereby proving the claim. 
\end{proof}

Finally, we prove Theorem \ref{thm1}.
\begin{proof}[Proof of Theorem \ref{thm1}]
	
    ($\Longleftarrow$) This direction is straightforward. Suppose, for the sake of contradiction, that this direction does not hold, i.e., \[-1\in\operatorname{SOS}+\mathcal{L}(\mathcal{N})+\mathcal{L}(\mathcal{N})^*\]
    and there exists a $*-$representation $ \rho $ 
	such that 
    \[\rho(\mathcal{N})=\{0\},\] then we have
	\begin{equation*}
	-1=\rho(-1)\in\rho(\operatorname{SOS}_{\mathcal{A}})\geq 0,
	\end{equation*}
	which is a contradiction!
	
	($\Longrightarrow$) This follows  from  Proposition \ref{p1} and Proposition \ref{p2}.
\end{proof}

\section{Some Discussions}
Here are some remarks and discussions about our results.
\begin{remark}
    Watts, Helton, and Klep demonstrated that for a torically determined game, the question of whether the game has a perfect commuting operator strategy can be translated into a subgroup membership problem \cite[Section 5]{watts2021noncommutative}. However, this result cannot be used to prove our theorem. The reason is that if we regard  $\mathcal{N}$ as the determining set of the game, the elements in $\mathcal{N}$ may not necessarily be expressible in the form $\beta g-1,\beta\in\bC,g\in G$. In other words, a two-answer game is not necessarily a torically determined game. 
\end{remark}
\begin{remark}
    Suppose the answer set $ A $ or $ B $ contains three or more elements. In that case, it is well known that our main result (Theorem \ref{thm1}) fails to hold, as there exists a nonlocal game that has a perfect commuting operator strategy but no perfect classical strategies\cite{cleve2004consequences,slofstra2019set}. From another perspective, equation (\ref{eq1}) no longer holds in this case, which prevents us from reaching a similar conclusion.
\end{remark}
\begin{remark}
    The algebra $\mathcal{A}$ is finitely generated, and the set $\mathcal{N}$ is also finite. However, the proof of our theorem relies on infinite-dimensional space.  It is currently unclear whether the proof can be simplified to avoid the use of infinite-dimensional spaces. 
\end{remark}

\begin{acks}
The authors would like to thank Sizhuo Yan, Jianting Yang, and Yuming Zhao for their helpful discussions and suggestions. 
\end{acks}
\bibliographystyle{ACM-Reference-Format}
\bibliography{LYZ2024ref.bib}


\begin{thebibliography}{17}


\ifx \showCODEN    \undefined \def \showCODEN     #1{\unskip}     \fi
\ifx \showDOI      \undefined \def \showDOI       #1{#1}\fi
\ifx \showISBNx    \undefined \def \showISBNx     #1{\unskip}     \fi
\ifx \showISBNxiii \undefined \def \showISBNxiii  #1{\unskip}     \fi
\ifx \showISSN     \undefined \def \showISSN      #1{\unskip}     \fi
\ifx \showLCCN     \undefined \def \showLCCN      #1{\unskip}     \fi
\ifx \shownote     \undefined \def \shownote      #1{#1}          \fi
\ifx \showarticletitle \undefined \def \showarticletitle #1{#1}   \fi
\ifx \showURL      \undefined \def \showURL       {\relax}        \fi
\providecommand\bibfield[2]{#2}
\providecommand\bibinfo[2]{#2}
\providecommand\natexlab[1]{#1}
\providecommand\showeprint[2][]{arXiv:#2}

\bibitem[Barvinok(2002)]%
        {barvinok2002course}
\bibfield{author}{\bibinfo{person}{Alexander Barvinok}.}
  \bibinfo{year}{2002}\natexlab{}.
\newblock \bibinfo{booktitle}{\emph{A course in convexity}}.
  Vol.~\bibinfo{volume}{54}.
\newblock \bibinfo{publisher}{American Mathematical Soc.}
\newblock


\bibitem[Bell(1964)]%
        {bell1964einstein}
\bibfield{author}{\bibinfo{person}{John~S Bell}.}
  \bibinfo{year}{1964}\natexlab{}.
\newblock \showarticletitle{On the einstein podolsky rosen paradox}.
\newblock \bibinfo{journal}{\emph{Physics Physique Fizika}}
  \bibinfo{volume}{1}, \bibinfo{number}{3} (\bibinfo{year}{1964}),
  \bibinfo{pages}{195}.
\newblock


\bibitem[Bene~Watts et~al\mbox{.}(2023)]%
        {watts2021noncommutative}
\bibfield{author}{\bibinfo{person}{Adam Bene~Watts}, \bibinfo{person}{J~William
  Helton}, {and} \bibinfo{person}{Igor Klep}.} \bibinfo{year}{2023}\natexlab{}.
\newblock \showarticletitle{Noncommutative Nullstellens{\"a}tze and Perfect
  Games}. In \bibinfo{booktitle}{\emph{Annales Henri Poincar{\'e}}},
  Vol.~\bibinfo{volume}{24}. Springer, \bibinfo{pages}{2183--2239}.
\newblock


\bibitem[Cimpri{\v{c}}(2009)]%
        {cimprivc2009representation}
\bibfield{author}{\bibinfo{person}{Jakob Cimpri{\v{c}}}.}
  \bibinfo{year}{2009}\natexlab{}.
\newblock \showarticletitle{A representation theorem for Archimedean quadratic
  modules on $*-$rings}.
\newblock \bibinfo{journal}{\emph{Canad. Math. Bull.}} \bibinfo{volume}{52},
  \bibinfo{number}{1} (\bibinfo{year}{2009}), \bibinfo{pages}{39--52}.
\newblock


\bibitem[Cimpri{\v{c}} et~al\mbox{.}(2014)]%
        {cimprivc2014real}
\bibfield{author}{\bibinfo{person}{Jakob Cimpri{\v{c}}},
  \bibinfo{person}{J~William Helton}, \bibinfo{person}{Igor Klep},
  \bibinfo{person}{Scott McCullough}, {and} \bibinfo{person}{Christopher
  Nelson}.} \bibinfo{year}{2014}\natexlab{}.
\newblock \showarticletitle{On real one-sided ideals in a free algebra}.
\newblock \bibinfo{journal}{\emph{Journal of Pure and Applied Algebra}}
  \bibinfo{volume}{218}, \bibinfo{number}{2} (\bibinfo{year}{2014}),
  \bibinfo{pages}{269--284}.
\newblock


\bibitem[Cimpri{\v{c}} et~al\mbox{.}(2013)]%
        {cimprivc2013noncommutative}
\bibfield{author}{\bibinfo{person}{Jakob Cimpri{\v{c}}},
  \bibinfo{person}{J~William Helton}, \bibinfo{person}{Scott McCullough}, {and}
  \bibinfo{person}{Christopher Nelson}.} \bibinfo{year}{2013}\natexlab{}.
\newblock \showarticletitle{A noncommutative real nullstellensatz corresponds
  to a noncommutative real ideal: Algorithms}.
\newblock \bibinfo{journal}{\emph{Proceedings of the London Mathematical
  Society}} \bibinfo{volume}{106}, \bibinfo{number}{5} (\bibinfo{year}{2013}),
  \bibinfo{pages}{1060--1086}.
\newblock


\bibitem[Cimpri\v{c} et~al\mbox{.}(2015)]%
        {cimpric2015real}
\bibfield{author}{\bibinfo{person}{J Cimpri\v{c}}, \bibinfo{person}{J Helton},
  \bibinfo{person}{S McCullough}, {and} \bibinfo{person}{C Nelson}.}
  \bibinfo{year}{2015}\natexlab{}.
\newblock \showarticletitle{Real nullstellensatz and*-ideals in*-algebras}.
\newblock \bibinfo{journal}{\emph{The Electronic Journal of Linear Algebra}}
  \bibinfo{volume}{30} (\bibinfo{year}{2015}), \bibinfo{pages}{19--50}.
\newblock


\bibitem[Clauser et~al\mbox{.}(1969)]%
        {clauser1969proposed}
\bibfield{author}{\bibinfo{person}{John~F Clauser}, \bibinfo{person}{Michael~A
  Horne}, \bibinfo{person}{Abner Shimony}, {and} \bibinfo{person}{Richard~A
  Holt}.} \bibinfo{year}{1969}\natexlab{}.
\newblock \showarticletitle{Proposed experiment to test local hidden-variable
  theories}.
\newblock \bibinfo{journal}{\emph{Physical review letters}}
  \bibinfo{volume}{23}, \bibinfo{number}{15} (\bibinfo{year}{1969}),
  \bibinfo{pages}{880}.
\newblock


\bibitem[Cleve et~al\mbox{.}(2004)]%
        {cleve2004consequences}
\bibfield{author}{\bibinfo{person}{Richard Cleve}, \bibinfo{person}{Peter
  Hoyer}, \bibinfo{person}{Benjamin Toner}, {and} \bibinfo{person}{John
  Watrous}.} \bibinfo{year}{2004}\natexlab{}.
\newblock \showarticletitle{Consequences and limits of nonlocal strategies}. In
  \bibinfo{booktitle}{\emph{Proceedings. 19th IEEE Annual Conference on
  Computational Complexity, 2004.}} IEEE, \bibinfo{pages}{236--249}.
\newblock


\bibitem[Dykema et~al\mbox{.}(2019)]%
        {dykema2019non}
\bibfield{author}{\bibinfo{person}{Ken Dykema}, \bibinfo{person}{Vern~I
  Paulsen}, {and} \bibinfo{person}{Jitendra Prakash}.}
  \bibinfo{year}{2019}\natexlab{}.
\newblock \showarticletitle{Non-closure of the set of quantum correlations via
  graphs}.
\newblock \bibinfo{journal}{\emph{Communications in Mathematical Physics}}
  \bibinfo{volume}{365} (\bibinfo{year}{2019}), \bibinfo{pages}{1125--1142}.
\newblock


\bibitem[Helton(2002)]%
        {H2002}
\bibfield{author}{\bibinfo{person}{J.~William Helton}.}
  \bibinfo{year}{2002}\natexlab{}.
\newblock \showarticletitle{``{Positive}" Noncommutative Polynomials Are Sums
  of Squares}.
\newblock \bibinfo{journal}{\emph{Annals of Mathematics}}
  \bibinfo{volume}{156}, \bibinfo{number}{2} (\bibinfo{year}{2002}),
  \bibinfo{pages}{675--694}.
\newblock


\bibitem[Helton and Mccullough(2004)]%
        {HMS2004}
\bibfield{author}{\bibinfo{person}{J.~William Helton} {and}
  \bibinfo{person}{Scott Mccullough}.} \bibinfo{year}{2004}\natexlab{}.
\newblock \showarticletitle{A {Positivstellensatz} for Non-commutative
  Polynomials}.
\newblock \bibinfo{journal}{\emph{Trans. Amer. Math. Soc.}}
  \bibinfo{volume}{356}, \bibinfo{number}{9} (\bibinfo{date}{01}
  \bibinfo{year}{2004}), \bibinfo{pages}{3721--3737}.
\newblock


\bibitem[Lupini et~al\mbox{.}(2020)]%
        {lupini2020perfect}
\bibfield{author}{\bibinfo{person}{Martino Lupini}, \bibinfo{person}{Laura
  Man{\v{c}}inska}, \bibinfo{person}{Vern~I Paulsen}, \bibinfo{person}{David~E
  Roberson}, \bibinfo{person}{Giannicola Scarpa}, \bibinfo{person}{Simone
  Severini}, \bibinfo{person}{Ivan~G Todorov}, {and} \bibinfo{person}{Andreas
  Winter}.} \bibinfo{year}{2020}\natexlab{}.
\newblock \showarticletitle{Perfect strategies for non-local games}.
\newblock \bibinfo{journal}{\emph{Mathematical Physics, Analysis and Geometry}}
  \bibinfo{volume}{23}, \bibinfo{number}{1} (\bibinfo{year}{2020}),
  \bibinfo{pages}{7}.
\newblock


\bibitem[Netzer and Thom(2013)]%
        {netzer2013real}
\bibfield{author}{\bibinfo{person}{Tim Netzer} {and} \bibinfo{person}{Andreas
  Thom}.} \bibinfo{year}{2013}\natexlab{}.
\newblock \showarticletitle{Real closed separation theorems and applications to
  group algebras}.
\newblock \bibinfo{journal}{\emph{Pacific J. Math.}} \bibinfo{volume}{263},
  \bibinfo{number}{2} (\bibinfo{year}{2013}), \bibinfo{pages}{435--452}.
\newblock


\bibitem[Slofstra(2019)]%
        {slofstra2019set}
\bibfield{author}{\bibinfo{person}{William Slofstra}.}
  \bibinfo{year}{2019}\natexlab{}.
\newblock \showarticletitle{The set of quantum correlations is not closed}. In
  \bibinfo{booktitle}{\emph{Forum of Mathematics, Pi}},
  Vol.~\bibinfo{volume}{7}. Cambridge University Press, \bibinfo{pages}{e1}.
\newblock


\bibitem[Watts and Helton(2020)]%
        {watts20203xor}
\bibfield{author}{\bibinfo{person}{Adam~Bene Watts} {and}
  \bibinfo{person}{J~William Helton}.} \bibinfo{year}{2020}\natexlab{}.
\newblock \showarticletitle{3XOR Games with Perfect Commuting Operator
  Strategies Have Perfect Tensor Product Strategies and are Decidable in
  Polynomial Time}.
\newblock \bibinfo{journal}{\emph{arXiv preprint arXiv:2010.16290}}
  (\bibinfo{year}{2020}).
\newblock


\bibitem[Yan et~al\mbox{.}(2023)]%
        {yan2023characterization}
\bibfield{author}{\bibinfo{person}{Sizhuo Yan}, \bibinfo{person}{Jianting
  Yang}, \bibinfo{person}{Tianshi Yu}, {and} \bibinfo{person}{Lihong Zhi}.}
  \bibinfo{year}{2023}\natexlab{}.
\newblock \showarticletitle{A Characterization of Perfect Strategies for Mirror
  Games}. In \bibinfo{booktitle}{\emph{Proceedings of the 2023 International
  Symposium on Symbolic and Algebraic Computation}}. \bibinfo{pages}{545--554}.
\newblock


\end{thebibliography}

\end{document}